\documentclass[12pt]{article}
\usepackage{amsmath}
\usepackage{inputenc}
\usepackage{graphicx}
\usepackage{mathtools}
\usepackage{amsthm}
\usepackage[hyperfootnotes=false]{hyperref}
\usepackage{epsfig,epic,eepic,units}
\usepackage{hyperref}
\usepackage{url}
\usepackage{longtable}
\usepackage{mathrsfs}
\usepackage{multirow}
\usepackage{bigstrut}
\usepackage{amssymb}
\usepackage{caption}
\usepackage{setspace}
\usepackage[margin=1.25in]{geometry}
\usepackage [english]{babel}
\usepackage [autostyle, english = american]{csquotes}
\MakeOuterQuote{"}

\newtheorem{theorem}{Theorem}

\begin{document}
\begin{singlespacing}
\title{\textbf{Exact Replication of the Best Rebalancing Rule in Hindsight\footnote{I thank anonymous reviewers for their time, effort, and valuable comments that improved the paper.}}}
\author{Alex Garivaltis\footnote{Assistant Professor of Economics, Northern Illinois University, 514 Zulauf Hall, DeKalb IL 60115.  E-mail:  agarivaltis1@niu.edu.  ORCID:  0000-0003-0944-8517.}}
\maketitle
\abstract{This paper prices and replicates the financial derivative whose payoff at $T$ is the wealth that would have accrued to a $\$1$ deposit into the best continuously-rebalanced portfolio (or fixed-fraction betting scheme) determined in hindsight.  For the single-stock Black-Scholes market, Ordentlich and Cover (1998) only priced this derivative at time-0, giving $C_0=1+\sigma\sqrt{T/(2\pi)}$.  Of course, the general time-$t$ price is \textit{not} equal to $1+\sigma\sqrt{(T-t)/(2\pi)}$.
\par
I complete the Ordentlich-Cover (1998) analysis by deriving the price at any time $t$.  By contrast, I also study the more natural case of the best \textit{levered} rebalancing rule in hindsight.  This yields $C(S,t)=\sqrt{T/t}\cdot\,\exp\{rt+\sigma^2b(S,t)^2\cdot t/2\}$, where $b(S,t)$ is the best rebalancing rule in hindsight over the observed history $[0,t]$.  I show that the replicating strategy amounts to  betting the fraction $b(S,t)$ of wealth on the stock over the interval $[t,t+dt].$  This fact holds for the general market with $n$ correlated stocks in geometric Brownian motion:  we get $C(S,t)=(T/t)^{n/2}\exp(rt+b'\Sigma b\cdot t/2)$, where $\Sigma$ is the covariance of instantaneous returns per unit time.  This result matches the $\mathcal{O}(T^{n/2})$ ``cost of universality'' derived by Cover in his ``universal portfolio theory'' (1986, 1991, 1996, 1998), which super-replicates the same derivative in discrete-time.  The replicating strategy compounds its money at the same asymptotic rate as the best levered rebalancing rule in hindsight, thereby beating the market asymptotically.  Naturally enough, we find that the American-style version of Cover's Derivative is never exercised early in equilibrium.
}
\end{singlespacing}
\titlepage
\section{Introduction}
The exotic option literature has several examples (Wilmott 1998) of derivatives with ``lookback'' or "no-regret" features.  For example, a floating-strike lookback call allows its owner to look back at the price history of a given stock, buy a share at the realized minimum $m:=\underset{1\leq t\leq T}{\min}\,\,S_t$, and sell it at the terminal price $S_T$.  Similarly, a fixed-strike lookback call allows its owner to buy one share at a fixed price $K$, and sell it at the historical maximum $M:=\underset{1\leq t\leq T}{\max}\,\,S_t$.  
\par
This paper prices and replicates a markedly different type of lookback option, whose payoff is equal to the final wealth that would have accrued to a $\$1$ deposit into the best continuous rebalancing rule (or fixed-fraction betting scheme) determined in hindsight.  This contingent claim has been studied by Cover and his collaborators (1986, 1991, 1996, 1998) who used it as a performance benchmark for discrete-time portfolio selection algorithms.  Ordentlich and Cover's important (1998) paper (on the ``max-min universal portfolio'') super-replicates this derivative in discrete-time. 
\par
In the context of one underlying stock, a rebalancing rule is a fixed-fraction betting scheme that continuously maintains some fraction $b\in(-\infty,+\infty)$ of wealth in the stock and keeps the rest in cash.  The portfolio is held for the differential time interval $[t,t+dt]$, at which point it is rebalanced to the target allocation.  If $b>1$, the scheme uses margin loans, but continuously maintains a fixed debt-to-assets ratio of $1-1/b$.  Say, for $b=2$ the scheme would keep a $50\%$ loan-to-value ratio at all times.  Thus, when the stock rises, the trader instantly adjusts by borrowing additional cash against his new wealth.  Similarly, on a downtick he will de-lever himself by selling a precise amount of the stock.  For example, using $b=2$ on the S\&P 500 index from January 2012 through August 2018 would have, under monthly rebalancing, compounded one's money at $31.8\%$ annually, as compared to buying and holding the index ($b=1$), which would have yielded $15.6\%$ annually.  This is illustrated in Figure \ref{fig:index}. 
\begin{figure}[t]
\centering
\includegraphics[width=350px]{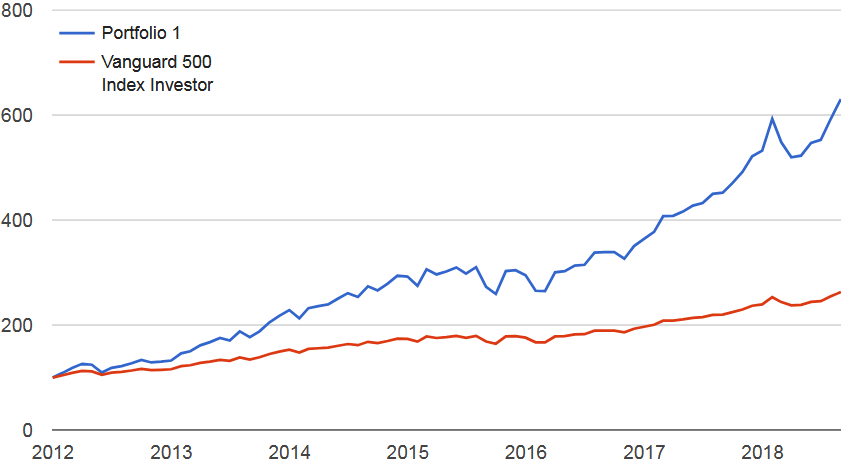}
\caption{\sc $b=2$ for Vanguard S\&P 500 index ETF under monthly rebalancing, Jan 2012-Aug 2018.}
\label{fig:index}
\end{figure}
\par
By contrast to the constant leveraged (2x) exposure discussed above, rebalancing rules $b\in(0,1)$ amount to ``volatility harvesting'' strategies (Luenberger 1998) that ``live off the fluctuations'' of the underlying.  Such rules are mechanical schemes for ``buying the dips and selling the rips,'' and they profit from mean-reversion in cyclical or ``sideways'' markets.  For example, using $b=0.5$ for shares of Advanced Micro Devices (AMD) with monthly rebalancing over the author's lifetime (April 1986 through August 2018), the trader would have compounded at $7.79\%$ per year, compared to $1.77\%$ for $b=1$.  This is illustrated in Figure \ref{fig:amd}.
\begin{figure}[t]
\centering
\includegraphics[width=350px]{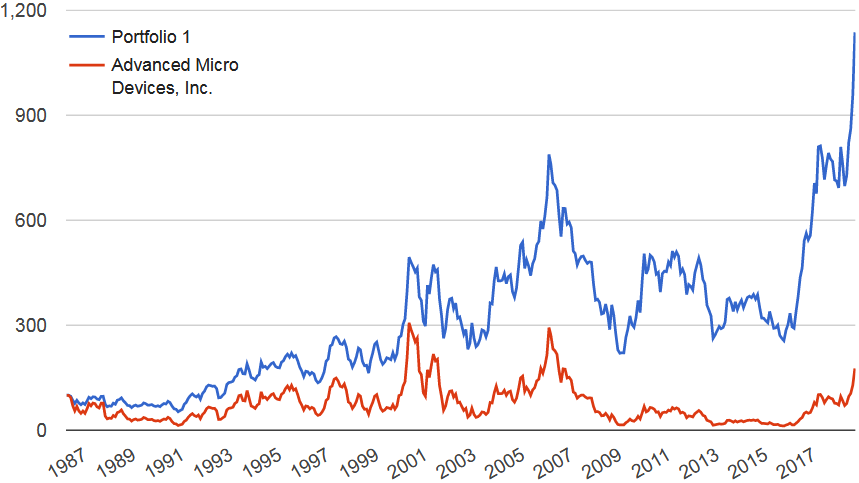}
\caption{\sc $b=0.5$ for AMD shares under monthly rebalancing, Apr 1986-Aug 2018.}
\label{fig:amd}
\end{figure}
\par
These examples make it clear that the best  rebalancing rule in hindsight will handily outperform the underlying over long periods.  For an underlying whose price failed to rise during the lookback period, the best rebalancing rule in hindsight can outperform by holding all cash ($b=0$) or by shorting the stock ($b<0$).  Inevitably, one lives to regret the fact that he did not use the best rebalancing rule in hindsight.  In 1986, no one could have reliably predicted that $b=0.5$ would beat AMD by 6 percent a year.  But (at least in the Black-Scholes world) it \textit{was} possible to delta-hedge the final wealth of the best continuous rebalancing rule in hindsight.  Such is the business of this paper.
\subsection{Contribution}
Ordentlich and Cover (1998) priced this derivative at time-0, for a single underlying with unlevered hindsight optimization.  The last result in their paper is the formula $C_0=1+\sigma\sqrt{T/(2\pi)}$, where $T$ is the horizon and $\sigma$ is the volatility.  Of course, the general time-$t$ price is \textit{not} equal to $1+\sigma\sqrt{(T-t)/(2\pi)}$.  Accordingly, this paper completes the Ordentlich-Cover (1998) analysis, deriving the eponymous Cost of Achieving the Best (constant-rebalanced) Portfolio in Hindsight at any time $t$, for levered hindsight optimization over any number of correlated stocks in geometric Brownian motion.  When leverage is allowed in the hindsight optimization, replication becomes especially simple.  At time $t$, we just look back at the observed history $[0,t]$ and compute the best (currently known) rebalancing rule in hindsight, here denoted $b(S,t)$.  We then bet the fraction $b(S,t)$ of wealth on the stock over $[t,t+dt]$.  This is equivalent to holding $\Delta(S,t):=b(S,t)C(S,t)/S$ shares of the stock in state $(S,t)$.  The replicating strategy serves to translate Ordentlich and Cover's (1998) ``max-min universal portfolio'' into continuous time.  Thus, the present paper does for Ordentlich and Cover (1998) what Jamshidian (1992) did for Cover's original (1991) performance-weighted universal portfolio.
\subsection{Related Literature}
Cover's universal portfolios (and the corresponding \textit{individual sequence approach} to investment) have generated a thriving literature in mathematical finance, computer science, and machine learning.  Parkes and Huberman (2001) study a cooperative multiagent search model of portfolio selection.  Barron and Yu (2003) supply investment strategies that are universal with respect to constant-rebalanced option portfolios.  Iyengar (2005) analyzes universal growth-optimal investment in a discrete-time market with transaction costs.  Stoltz and Lugosi (2005) extend the concept of internal regret to the on-line portfolio selection problem.  They develop sequential investment strategies that minimize cumulative internal regret under model uncertainty.  DeMarzo, Kremer, and Mansour (2006) use on-line trading algorithms and regret minimization to derive robust bounds for option prices.  Gy\"{o}rfi, Lugosi, and Udina (2006) study kernel-based sequential investment strategies that guarantee optimal capital growth for markets with ergodic stationary  return processes.  Kozat and Singer (2011) deal with semiconstant-rebalanced portfolios that rebalance only at selected points in time, and thus may avoid rebalancing if the prospective benefits are outweighed by transaction costs.  They exhibit on-line investment strategies that asymptotically achieve the wealth of the best semiconstant-rebalanced portfolio for the realized sequence of asset returns.  Rujeerapaiboon, Kuhn, and Wiesemann (2015) use robust optimization techniques to build fixed-mix strategies offering performance guarantees that are similar to the growth-optimal portfolio.
\par
Portfolio rebalancing is a key tenet of Fernholz's (2002) stochastic portfolio theory.  Wong (2015) extends Cover's universal portfolio to the nonparametric family of functionally generated portfolios in stochastic portfolio theory.  Cuchiero, Schachermayer, and Wong (2016) show that, under appropriate hypotheses, the asymptotic compound-growth rate of Cover's universal portfolio coincides with that of stochastic portfolio theory and the num\'eraire portfolio.  On a more practical basis, rebalancing is a perennially important aspect of tactical asset allocation.  For instance, Israelov and Tummala (2018) study short option overlays that can be used to hedge one's exposure to allocation drift between planned rebalances.  Gort and Burgener (2014) describe and backtest option selling techniques that serve to rebalance institutional investors' asset exposures to predefined targets.  An AQR White Paper by Ilmanen and Maloney (2015) discusses the key considerations for investors deciding on whether and how to rebalance their portfolios.
\section{One Underlying}
\subsection{Payoff Computation}
For simplicity, we start with a single underlying stock whose price $S_t$ follows the geometric Brownian motion
\begin{equation}
\frac{dS_t}{S_t}=\mu\,dt+\sigma\,dW_t,
\end{equation}where $\mu$ is the drift, $\sigma$ is the volatility, and $W_t$ is a standard Brownian motion.  There is a risk-free bond whose price $B_t:=e^{rt}$ follows
\begin{equation}
\frac{dB_t}{B_t}=r\,dt,
\end{equation}  where $r$ is the continuously-compounded interest rate.  We consider constant rebalancing rules, or fixed-fraction betting schemes, that ``bet'' the fraction $b\in(-\infty,+\infty)$ of wealth on the stock over the interval $[t,t+dt]$.  Assume that the gambler starts with $\$1$, and let $V_t=V_t(b)$ denote his wealth at $t$.  He thus owns $bV_t/S_t$ shares of the stock at $t$, and has the remaining $(1-b)V_t$ dollars invested in bonds.  The gambler's wealth evolves according to
\begin{equation}
\frac{dV_t(b)}{V_t(b)}=b\,\frac{dS_t}{S_t}+(1-b)\frac{dB_t}{B_t}=[r+b(\mu-r)]dt+b\sigma\, dW_t.
\end{equation}
Since $V_t(b)$ is a geometric Brownian motion, we have 
\begin{equation}\label{fortune}
V_t(b)=\exp\{[r+(\mu-r)b-\sigma^2b^2/2]t+b\sigma W_t\}.
\end{equation}In the formula
\begin{equation}
S_t=S_0\exp\{(\mu-\sigma^2/2)t+\sigma W_t\},
\end{equation}we can solve for $\sigma W_t$ in terms of $S_t$, and substitute the resulting expression into (\ref{fortune}).  This yields
\begin{equation}
V_t(b)=\exp\{(r-\sigma^2b^2/2)t+b[\log(S_t/S_0)-(r-\sigma^2/2)t]\}.  
\end{equation}  Thus, we note that $V_t(b)$ can be calculated without any explicit reference to the drift parameter $\mu$. The trader's wealth is \textit{Markovian}: it depends only on the current state $(S_t,t)$. 
\par
To find the best rebalancing rule in hindsight over $[0,t]$, we maximize $V_t(b)$ with respect to $b$.  Since the exponent is quadratic in $b$, the best rebalancing rule in hindsight over $[0,t]$ is
\begin{equation}
\boxed{b(S_t,t):=\frac{\log(S_t/S_0)-(r-\sigma^2/2)t}{\sigma^2t}}.
\end{equation}  If we write $\hat{\mu}(S,t):=\log(S/S_0)/t+\sigma^2/2$, we get the expression
\begin{equation}
\boxed{b(S,t)=\frac{\hat{\mu}(S,t)-r}{\sigma^2}}.
\end{equation}  Let $V_t^*:=\max\limits_{b\in\mathbb{R}}V_t(b)$ denote the final wealth of the best levered rebalancing rule in hindsight over $[0,t]$.  Then
\begin{equation}
\boxed{V_t^*=\exp\{rt+\sigma^2b(S,t)^2\cdot t/2\}}.
\end{equation}
Figure \ref{fig:payoff} plots this payoff for different volatilities, assuming a risk-free rate of $r:=0$ over a horizon of $T:=5$ years.
\begin{figure}[t]
\centering
\includegraphics[width=430px]{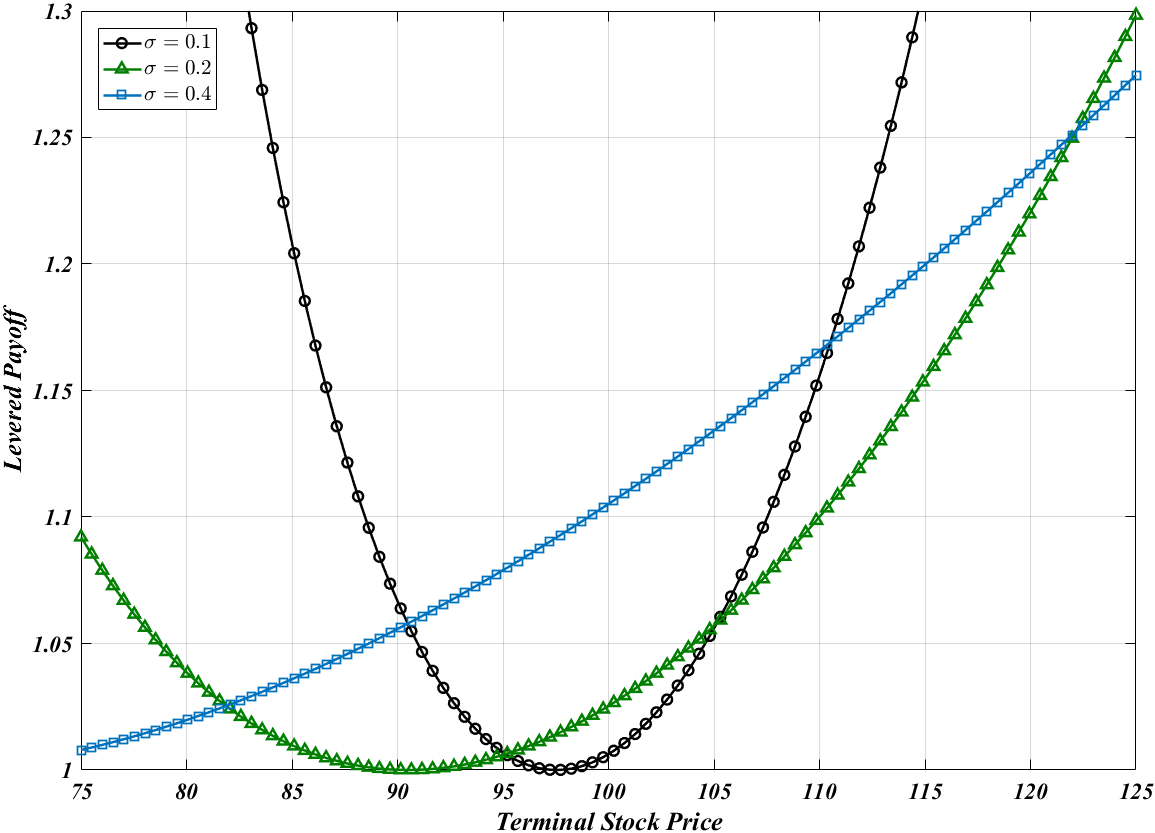}
\caption{\sc The payoff of Cover's Derivative for levered hindsight optimization, $T:=5, S_0:=100, r:=0.$}
\label{fig:payoff}
\end{figure}
In Ordentlich and Cover (1998), the hindsight optimization is over unlevered rebalancing rules $b\in[0,1]$, and in that context, the best unlevered rebalancing rule in hindsight is $b^u(S,t):=\max\{\min\{b(S,t),1\},0\}$.  Thus, they use the payoff
\begin{align}\boxed{
V_t^*:=\max_{0\leq b\leq 1}V_t(b)=\left\{ \begin{array}{cc} 
                e^{rt} & \hspace{5mm} \text{if}\,\,b(S,t)\leq0\\
                \exp\{rt+\sigma^2b(S,t)^2\cdot t/2\} & \hspace{5mm} \text{if}\,\,0\leq b(S,t)\leq1\\
                S_t/S_0 & \text{if}\,\,b(S,t)\geq1 \\
                \end{array} \right.
}.\end{align}
Figure \ref{fig:unlev} plots the unlevered payoff for different volatilities, assuming a risk-free rate of $r:=0$ over a horizon of $T:=2$ years.
\begin{figure}[t]
\centering
\includegraphics[width=400px]{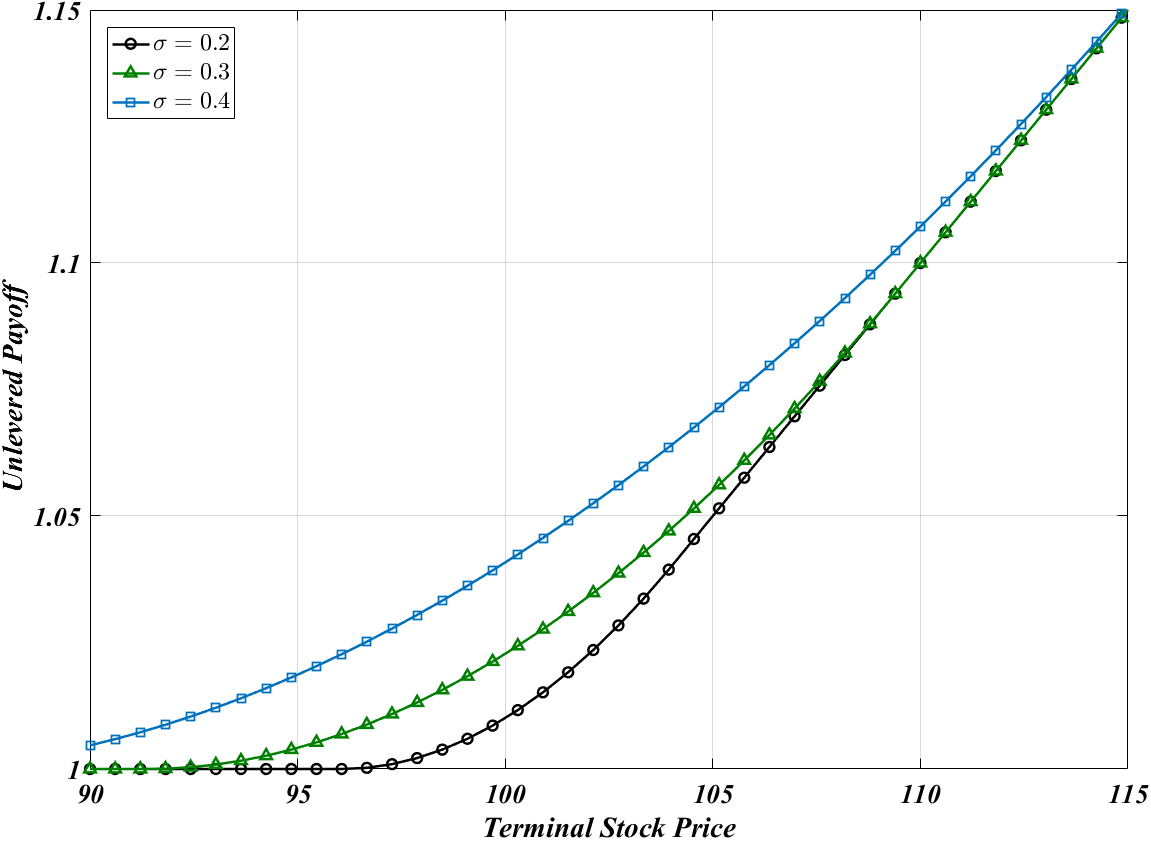}
\caption{\sc The payoff of Cover's Derivative for unlevered hindsight optimization, $T:=2, S_0:=100, r:=0.$}
\label{fig:unlev}
\end{figure}
Consider the European-style derivative (``hindsight allocation option'') whose payoff at $T$ is $V_T^*(S_T)$.
Let $\mathbb{Q}$ denote the equivalent martingale measure.  Ordentlich and Cover (1998) computed the expected present value \begin{equation}
\boxed{C_0:=e^{-rT}\mathbb{E}_0^{\mathbb{Q}}[V_T^*]=1+\sigma\sqrt{\frac{T}{2\pi}}}
\end{equation} with respect to $\mathbb{Q}$ and the information available at $t=0$.  If someone buys a dollar's worth of this derivative at $t=0$ (for some distant expiration date $T$), he will compound his money at the same asymptotic rate as the best unlevered rebalancing rule in hindsight.  His initial dollar buys him $1/C_0$ units of the derivative, yielding final wealth $V_T^*/\big\{1+\sigma\sqrt{T/(2\pi)}\big\}.$  After holding the option for $T$ years, the excess continuously-compounded growth rate of the best rebalancing rule in hindsight (over and above that of the option holder) is
\begin{equation}
\frac{1}{T}\log V^*_T-\frac{1}{T}\log\bigg\{\frac{V_T^*}{1+\sigma\sqrt{T/(2\pi)}}\bigg\}=\frac{1}{T}\log\bigg\{1+\sigma\sqrt{\frac{T}{2\pi}}\bigg\},
\end{equation}which tends to 0 as $T\to\infty$.  Note that the excess growth rate is \textit{deterministic}.  Figure \ref{fig:regret} plots this excess growth rate for different volatilities and maturities.
\begin{figure}[t]
\centering
\includegraphics[width=430px]{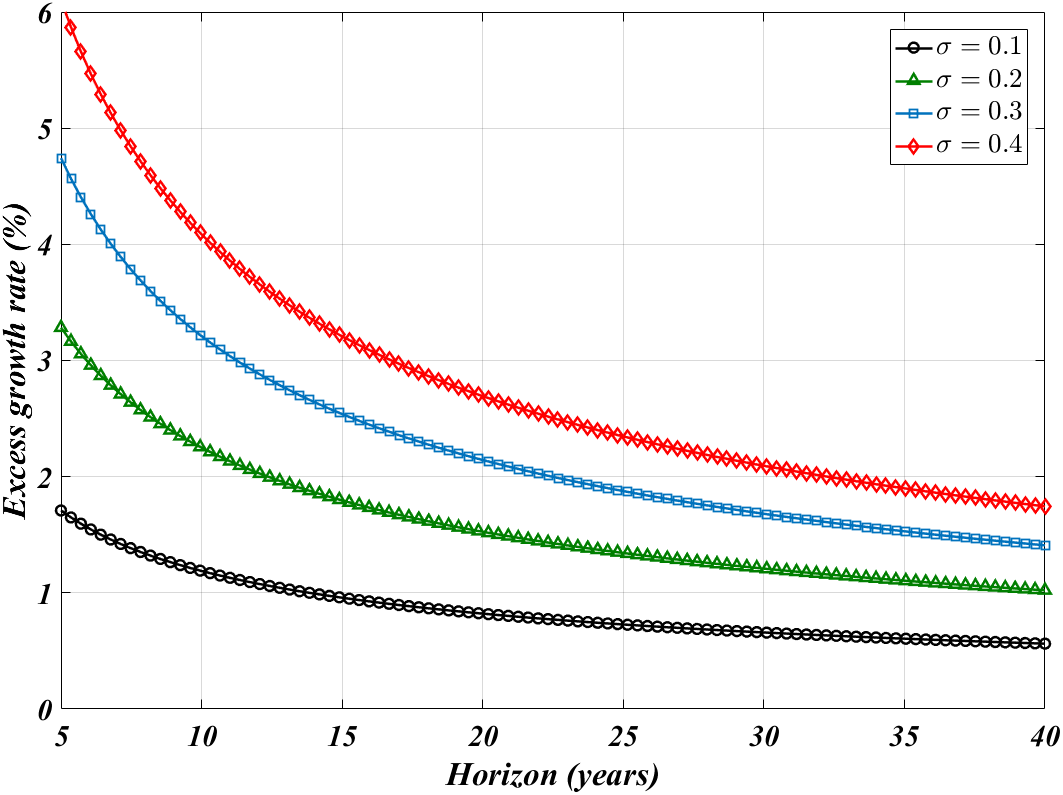}
\caption{\sc Excess continuously-compounded annual growth rate (\%) of the best (unlevered) rebalancing rule in hindsight over that of the replicating strategy.}
\label{fig:regret}
\end{figure}
\subsection{No-Arbitrage Price}
We find it somewhat more natural to start with levered hindsight optimization, corresponding to the payoff $V_T^*:=\max\limits_{b\in\mathbb{R}}V_T(b)$.  Accordingly, we take up the Black-Scholes (1973) equation
\begin{equation}
\frac{1}{2}\sigma^2S^2\frac{\partial^2 C}{\partial S^2}+rS\frac{\partial C}{\partial S}+\frac{\partial C}{\partial t}-rC=0
\end{equation} along with the boundary condition $C(S,T):=\exp\{rT+\sigma^2 b(S,T)^2\cdot T/2\}$.  For convenience, we define the auxiliary variable 
\begin{equation}
\boxed{z_t:=\frac{\log(S_t/S_0)-(r-\sigma^2/2)t}{\sigma\sqrt{t}}},
\end{equation}which is a unit normal with respect to the equivalent martingale measure.  Under this notation, we have
\begin{equation}
\boxed{b(S,t)=\frac{z(S,t)}{\sigma\sqrt{t}}}.
\end{equation}  Thus, the final payoff of Cover's rebalancing option is
\begin{equation}
V_T^*=\exp(rT+z_T^2/2).
\end{equation}  The intrinsic value at time $t$ is
\begin{equation}
\boxed{V_t^*=\exp(rt+z_t^2/2)}.
\end{equation}  We proceed to compute the expected discounted payoff with respect to the equivalent martingale measure and the information available at $t$.  To this end, we write
\begin{equation}
z_T=\sqrt{\frac{t}{T}}\cdot z_t+\sqrt{1-\frac{t}{T}}\cdot y,
\end{equation}where 
\begin{equation}
y:=\frac{\log(S_T/S_t)-(r-\sigma^2/2)(T-t)}{\sigma\sqrt{T-t}}.
\end{equation}  $y$ is a unit normal with respect to the equivalent martingale measure and the information available at $t$.  Thus, we have
\begin{equation}
\mathbb{E}_t^{\mathbb{Q}}[\exp(z_T^2/2)]=\frac{\exp\{tz_t^2/(2T)\}}{\sqrt{2\pi}}\int\limits_{-\infty}^{\infty}\exp\bigg\{-\frac{t}{2T}y^2+\frac{\sqrt{t(T-t)}}{T}z_ty\bigg\}dy.
\end{equation}  To evaluate the integral, we make note of the general formula (cf. the appendix to Reiner and Rubinstein 1992)
\begin{equation}
\int\limits_{A}^B e^{-\alpha y^2+\beta y}dy=\sqrt{\frac{\pi}{\alpha}}\exp\bigg(\frac{\beta^2}{4\alpha}\bigg)\bigg[N\bigg(B\sqrt{2\alpha}-\frac{\beta}{\sqrt{2\alpha}}\bigg)-N\bigg(A\sqrt{2\alpha}-\frac{\beta}{\sqrt{2\alpha}}\bigg)\bigg],
\end{equation}where $\alpha>0$ and $N(\cdot)$ is the cumulative normal distribution function.  Putting $\alpha:=t/(2T),$ $\beta:=\frac{\sqrt{t(T-t)}}{T}z_t$, $A:=-\infty$, and $B:=+\infty$, we get
\begin{equation}
\mathbb{E}_t^{\mathbb{Q}}[\exp(z_T^2/2)]=\sqrt{\frac{T}{t}}\exp(z_t^2/2).
\end{equation}
\begin{theorem}
For levered hindsight optimization (over all $b\in\mathbb{R}$), the price of Cover's rebalancing option is
\begin{equation}
\boxed{C(S,t)=\sqrt{\frac{T}{t}} \exp(rt+z^2/2)=\sqrt{\frac{T}{t}}\exp\{rt+\sigma^2b(S,t)^2\cdot t/2\}=\sqrt{\frac{T}{t}}V_t^*},
\end{equation}where $z:=\{\log(S/S_0)-(r-\sigma^2/2)t\}/(\sigma\sqrt{t})$, $b(S,t)$ is the best rebalancing rule in hindsight over $[0,t]$, and $V_t^*$ is the intrinsic value at time $t$.
\end{theorem}
\begin{theorem}
The American-style version of Cover's Derivative (that expires at $T$, has zero exercise price, and pays $V_t^*$ upon exercise at $t$) will never be excercised early in equilibrium.  The American price $C_a(S,t)$ is equal to the European price $C_e(S,t)=\sqrt{T/t}\cdot V_t^*$.
\end{theorem}
\begin{proof}
Note that $C_e(S_t,t)>V_t^*$ for $0<t<T$, e.g. the European price of Cover's Derivative always exceeds the exercise value.  To prevent arbitrage opportunities, we must have $C_a(S,t)\geq C_e(S,t)$ on account of the additional rights granted by the American-style option.  Thus, we always have $C_a(S_t,t)>V_t^*$, which means, to quote Merton's (1973) terminology, that the option is ``worth more alive than dead.''  In equilibrium, there are always willing buyers ready to pay more than the exercise value, so the option would be sold to such buyers instead of being exercised.  Thus, early exercise being useless anyhow, we conclude that $C_a(S,t)=C_e(S,t)$.
\end{proof}
To be quite formal about it, the present American option valuation problem (cf. Wilmott 1998) consists in solving the partial differential inequality
\begin{equation}
\frac{1}{2}\sigma^2S^2\frac{\partial^2C}{\partial S^2}+rS\frac{\partial C}{\partial S}+\frac{\partial C}{\partial t}-rC\leq0
\end{equation}together with the side conditions $C(S,T)=V_T^*, C(S,t)\geq V_t^*$, and subject to the proviso that $\partial C/\partial S$ is continuous.  These conditions are all indeed satisfied by the formula $C(S,t)=\sqrt{T/t}\cdot V_t^*=\sqrt{T/t}\cdot\exp\{rt+z(S,t)^2/2\}.$
\subsection{Replicating Strategy and the Greeks}
Differentiating the price, we find at once that 
\begin{equation}
\boxed{\Delta:=\frac{\partial C}{\partial S}=\frac{C\cdot z}{S\sigma\sqrt{t}}=\frac{C\cdot b(S,t)}{S}}.
\end{equation}
or, equivalently, that $\Delta S/C=b(S,t)$.
\begin{theorem}
The replicating strategy for Cover's Derivative bets the fraction $b(S,t)$ of wealth on the stock in state $(S,t)$.  Thus, to replicate Cover's Derivative, one just uses the best rebalancing rule in hindsight as it is known at time $t$.
\end{theorem}
Hence, for the complete market with a single stock in geometric Brownian motion, assuming levered hindsight optimization, the following three trading strategies are identical:
\begin{enumerate}
\item
The strategy that looks back over the known price history $[0,t]$, finds the best continuously-rebalanced portfolio in hindsight, and uses that portfolio over the interval $[t,t+dt]$
\item
The $\Delta$-hedging strategy induced by Cover's Derivative
\item
The natural estimator $(\hat{\mu}-r)/\sigma^2$ of the continuous-time Kelly rule $(\mu-r)/\sigma$ (cf. Luenberger 1998)
\end{enumerate}
For reference, we catalog the rest of the Greeks below.
\begin{equation}
\boxed{\Gamma:=\frac{\partial\Delta}{\partial S}=\frac{z^2-1}{(S\sigma\sqrt{t})^3}C}
\end{equation}
\begin{equation}
\boxed{\Theta:=\frac{\partial C}{\partial t}=\bigg(r-\frac{1}{2t}-\frac{z^2}{2}\bigg)C}
\end{equation}
Thus, there will be significant time decay in the option value for small times $t$ and for extreme price realizations in either direction.
\begin{equation}
\boxed{\nu:=\frac{\partial C}{\partial\sigma}=\bigg[\frac{\sqrt{t}}{2}+\frac{r+\log(S_0/S)}{\sigma^2t}\bigg]Cz}.
\end{equation}
There are generally two implied volatilities that rationalize a given observed value of $C$.  To show this, we start with the relation 
\begin{equation}\label{iveq}
z^2=2(\log\,C-rt)+\log(t/T).
\end{equation}Comparing (\ref{iveq}) with the definition
\begin{equation}
z^2=\frac{[\log(S_t/S_0)-(r-\sigma^2/2)t]^2}{\sigma^2t},
\end{equation}we get a quadratic equation in the variance $\sigma^2$.
The lowest possible rational option price is $\sqrt{T/t}\cdot e^{rt}$, which corresponds to $z_t=0$.  This happens if and when $S_t=S_0e^{(r-\sigma^2/2)t}$.  Figure \ref{fig:vol} plots the option price against $\sigma$ for the parameters $t:=0.5$, $T:=1$, $r:=0.03$, $S_0:=100$, and $S_t:=105$.
\begin{figure}[t]
\centering
\includegraphics[width=350px]{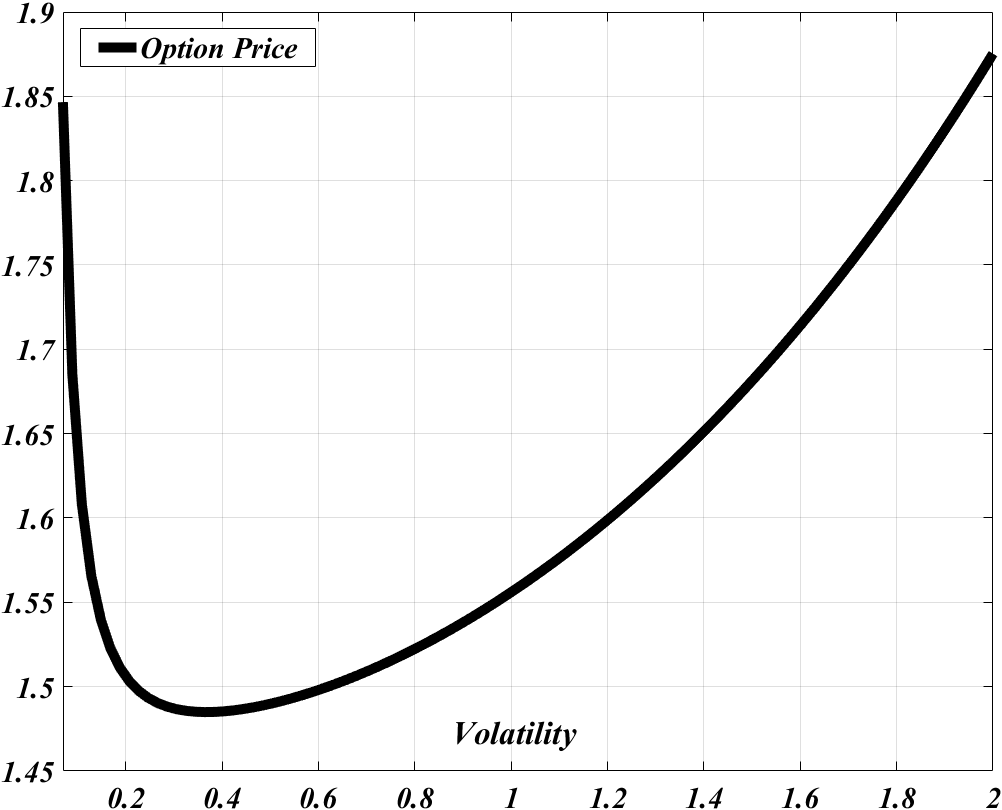}
\caption{\sc The dual implied volatilities that rationalize an observed price of Cover's Derivative,  $t:=0.5$, $T:=1$, $r:=0.03$, $S_0:=100$, $S_t:=105$.}
\label{fig:vol}
\end{figure}
Finally, we have the interest rate sensitivity
\begin{equation}
\boxed{\rho:=\frac{\partial C}{\partial r}=[1-b(S,t)]Ct}.
\end{equation}
Thus, when the best rebalancing rule in hindsight makes a positive allocation to cash, higher interest rates will make the option more valuable.  When the hindsight-optimized rebalancing rule uses margin debt ($b(S,t)>1$), higher interest rates will make the option less valuable.
\subsection{Unlevered Hindsight Optimization}
In this subsection, we take up the case of unlevered hindsight optimization, obtaining a more direct generalization of Ordentlich and Cover's (1998) formula $C_0=1+\sigma\sqrt{T/(2\pi)}$.  Thus, we consider the payoff
\begin{align}
V_t^*:=\left\{ \begin{array}{cc} 
                e^{rt} & \hspace{5mm} \text{if}\,\,z_t\leq0\\
                \exp(rt+z_t^2/2) & \hspace{5mm} \text{if}\,\,0\leq z_t\leq\sigma\sqrt{t},\\
                S_t/S_0 & \text{if}\,\,z_t\geq\sigma\sqrt{t} \\
                \end{array} \right.
\end{align}
where $z_t:=\{\log(S_t/S_0)-(r-\sigma^2/2)t\}/(\sigma\sqrt{t})$.
In this connection, the replicating strategy no longer coincides with the best (unlevered) rebalancing rule in hindsight $b^u(S,t)$ over the known history $[0,t]$.  
\par
Again, we make the decomposition $z_T=\sqrt{t/T}\cdot z_t+\sqrt{1-t/T}\cdot y$, where $y:=\{\log(S_T/S_t)-(r-\sigma^2/2)(T-t)\}/(\sigma\sqrt{T-t})$.  With this terminology, the final payoff becomes
\begin{align}
V_T^*=\left\{ \begin{array}{cc} 
                e^{rT} & \hspace{5mm} \text{if}\,\,y\leq-z_t\sqrt{t/(T-t)}\\
                \exp(rT+z_T^2/2) & \hspace{5mm} \text{if}\,\,-z_t\sqrt{t/(T-t)}\leq y\leq(\sigma T-\sqrt{t}z_t)/\sqrt{T-t}.\\
                S_T/S_0 & \text{if}\,\,y\geq(\sigma T-\sqrt{t}z_t)/\sqrt{T-t} \\
                \end{array} \right.
\end{align}
The expected discounted payoff is the sum of three integrals $I_1+I_2+I_3$, corresponding to the three events 
$b^u(S_T,T)=0$, $0<b^u(S_T,T)<1$, and $b^u(S_T,T)=1$.  Each integral constitutes a separate solution of the Black-Scholes equation.  To further simplify the notation, we define $A:=-z_t\sqrt{t/(T-t)}$ and $B:=A+\sigma T/\sqrt{T-t}$.  We have
\begin{equation}
I_1:=\frac{1}{\sqrt{2\pi}}\int\limits_{-\infty}^{A}\exp(rt-y^2/2)dy=e^{rt}N(A),
\end{equation}where $N(\cdot)$ is the cumulative normal distribution function.  Next, we get
\begin{equation}
I_2:=\frac{\exp\{rt+tz_t^2/(2T)\}}{\sqrt{2\pi}}\int\limits_{A}^{B}\exp\bigg(-\frac{t}{2T}y^2+\frac{\sqrt{t(T-t)}}{T}z_ty\bigg)dy.
\end{equation}Evaluating the integral and simplifying, one has
\begin{equation}
I_2=\sqrt{\frac{T}{t}}\exp(rt+z_t^2/2)\bigg[N\bigg(A\sqrt{\frac{T}{t}}+\sigma\sqrt{\frac{tT}{T-t}}\bigg)-N\bigg(A\sqrt{\frac{T}{t}}\bigg)\bigg].
\end{equation}
Finally, we calculate
\begin{equation}
I_3:=e^{-r(T-t)}\frac{S_t}{S_0}\cdot\frac{\exp\{(r-\sigma^2/2)(T-t)\}}{\sqrt{2\pi}}\int\limits_B^\infty \exp\big(-y^2/2+\sigma\sqrt{T-t}\cdot y\big)dy,
\end{equation}which simplifies to
\begin{equation}
I_3=\frac{S_t}{S_0}N(\sigma\sqrt{T-t}-B).
\end{equation}
\begin{theorem}
For the single-stock Black-Scholes market with unlevered hindsight optimization, the price $C^u(S,t)$ of Cover's Derivative is
\begin{multline}
C^{u}(S,t)=e^{rt}N(A)+C(S,t)\bigg[N\bigg(A\sqrt{\frac{T}{t}}+\sigma\sqrt{\frac{tT}{T-t}}\bigg)-N\bigg(A\sqrt{\frac{T}{t}}\bigg)\bigg]\\+\frac{S_t}{S_0}N(\sigma\sqrt{T-t}-B),\\
\end{multline}where $z:=\{\log(S_t/S_0)-(r-\sigma^2/2)t\}/(\sigma\sqrt{t})$, $A:=-z\sqrt{t/(T-t)}$, $B:=A+\sigma T/\sqrt{T-t}$, and $C(S,t):=\sqrt{T/t}\cdot \exp(rt+z^2/2)$ is the price of Cover's Derivative under levered hindsight optimization.
\end{theorem}
\subsection{Binomial Lattice Price}
For the sake of completeness, we proceed to derive the general price of Cover's Derivative on the Cox-Ross-Rubinstein (1979) binomial lattice.  By abuse of notation, let $r$ denote the per-period interest rate, with $R:=1+r$ being the gross rate of interest.  We subdivide the interval $[0,T]$ into $N$ subintervals of length $\Delta t:=T/N$.  The stock price $S(t)$ evolves according to 
\begin{equation}
S(t+\Delta t)=\begin{cases}
S(t)\cdot u& \text{with probability }p\\
S(t)\cdot d& \text{with probability }1-p\\
\end{cases}
\end{equation}where $u,d$ are constants such that $0<d<R<u$.  We let $q:=(R-d)/(u-d)$ be the risk-neutral probability, with $1-q=(u-R)/(u-d)$.  The payoff-relevant state is the number of ups $j$, where $0\leq j\leq N$.  After $N$ plays, the (possibly levered) rebalancing rule $b$ has grown the initial dollar into
\begin{equation}
V_T(b):=R^N[1+b(u/R-1)]^j[1+b(d/R-1)]^{N-j}.
\end{equation}  To get the best rebalancing rule in hindsight over $[0,T]$, we take logs and differentiate with respect to $b$, yielding the first-order condition
\begin{equation}
j(u-R)[1+b(d/R-1)]+(N-j)(d-R)[1+b(u/R-1)]=0.
\end{equation}  Solving and simplifying, the best rebalancing rule in hindsight (after $j$ ups and $N-j$ downs) is
\begin{equation}
b(j,N):=\frac{R}{N(u-d)}\bigg(\frac{j}{q}-\frac{N-j}{1-q}\bigg)
\end{equation}The final payoff of Cover's Derivative is
\begin{equation}
V_T^*(j,N):=\bigg(\frac{R}{N}\bigg)^N\bigg(\frac{j}{q}\bigg)^j\bigg(\frac{N-j}{1-q}\bigg)^{N-j},
\end{equation}  where we have adopted the convention that $0^0:=1$.  If the hindsight-optimization is restricted to unlevered rebalancing rules $b\in[0,1],$ then the payoff becomes
\begin{align}
V_T^*:=\left\{ \begin{array}{cc} 
                R^N & \hspace{5mm} \text{if }j\leq Nq\\
                (\frac{R}{N})^N\big(\frac{j}{q}\big)^j\big(\frac{N-j}{1-q}\big)^{N-j} & \hspace{5mm} \text{if }Nq<j<Nq+\frac{u-d}{Rq(1-q)}.\\
                u^jd^{N-j} & \text{if }j\geq Nq+\frac{u-d}{Rq(1-q)} \\
                \end{array} \right.
\end{align}
For unlevered hindsight optimization, Ordentlich and Cover (1998) gave us the time-0 price
\begin{multline}\label{oc98}
C_0=\text{Prob}\{j\leq Nq\}+\sum\limits_{Nq<j<Nq+\frac{u-d}{Rq(1-q)}}\binom{N}{j}\bigg(\frac{j}{N}\bigg)^j\bigg(1-\frac{j}{N}\bigg)^{N-j}\\+\sum\limits_{j\geq Nq+\frac{u-d}{Rq(1-q)}}\binom{N}{j}(qu)^j[(1-q)d]^{N-j}.\\
\end{multline}We supplement this formula by computing the general price under levered hindsight optimization in state $(k,n)$, where $k$ upticks have occured in the first $n$ time steps.  Letting $j$ denote the number of upticks in the next $N-n$ steps, the expected discounted payoff in state $(k,n)$ with respect to the risk-neutral measure is
\begin{equation}
C(k,n):=q^{-k}(1-q)^{k-n}\sum\limits_{j=0}^{N-n}\binom{N-n}{j}\bigg(\frac{j+k}{N}\bigg)^{j+k}\bigg(1-\frac{j+k}{N}\bigg)^{N-j-k}.
\end{equation}
This being done, one can replicate Cover's Derivative on the binomial lattice by using our general price $C(k,n)$ in conjunction with the formula
\begin{equation}
\Delta:=\frac{C_u-C_d}{S(u-d)}=\frac{C(k+1,n+1)-C(k,n+1)}{S(u-d)},
\end{equation}where $S$ is the current stock price, $n$ is the number of time steps to date, and $k$ is the number of upticks that have occured so far.
\par
To obtain a more direct generalization of (\ref{oc98}), we close this subsection by computing the price of Cover's Derivative for unlevered hindsight optimization in all possible states $(k,n).$  The price consists of three terms $C^u(k,n):=\Sigma_1+\Sigma_2+\Sigma_3$, corresponding to the three events $b^*\leq0$, $0<b^*<1$, and $b^*\geq 1.$  Again, $j$ will denote the number of upticks that occur over the next $N-n$ time steps.  We start with
\begin{equation}
\Sigma_1:=\sum\limits_{0\leq j\leq Nq-k}\binom{N-n}{j}q^j(1-q)^{N-n-j}.
\end{equation}Next, we get
\begin{equation}
\Sigma_2:=q^{-k}(1-q)^{k-n}\sum\limits_{Nq-k<j<Nq-k+\frac{u-d}{Rq(1-q)}}\binom{N-n}{j}	\bigg(\frac{j+k}{N}\bigg)^{j+k}\bigg(1-\frac{j+k}{N}\bigg)^{N-j-k}.
\end{equation}Finally, we have
\begin{equation}
\Sigma_3:=q^{-k}(1-q)^{k-n}\sum\limits_{j\geq Nq-k+\frac{u-d}{Rq(1-q)}}\binom{N-n}{j}(qu)^{k+j}[(1-q)d]^{N-k-j}.
\end{equation}
\subsubsection*{Simulation:  ``Shannon's Demon''}
To illustrate the replication of Cover's Derivative on a binomial lattice, we simulate Shannon's canonical discrete-time example (cf. Poundstone 2010).  This amounts to the parameters $u:=2$, $d:=1/2$, $r:=0$, $R=1$, and the risk-neutral probability $q=1/3$.  The gambler buys (replicates) a dollar's worth of Cover's Derivative at $n=0$, and holds the option until $n=N$.  His wealth after $n$ steps (and $k$ upticks) is $C(k,n)/C(0,0)$.  By comparison, the stock price will be $2^{2k-n}.$  Figure \ref{fig:demon} plots a sample path for $N:=300$ periods.
\begin{figure}[t]
\centering
\includegraphics[width=430px]{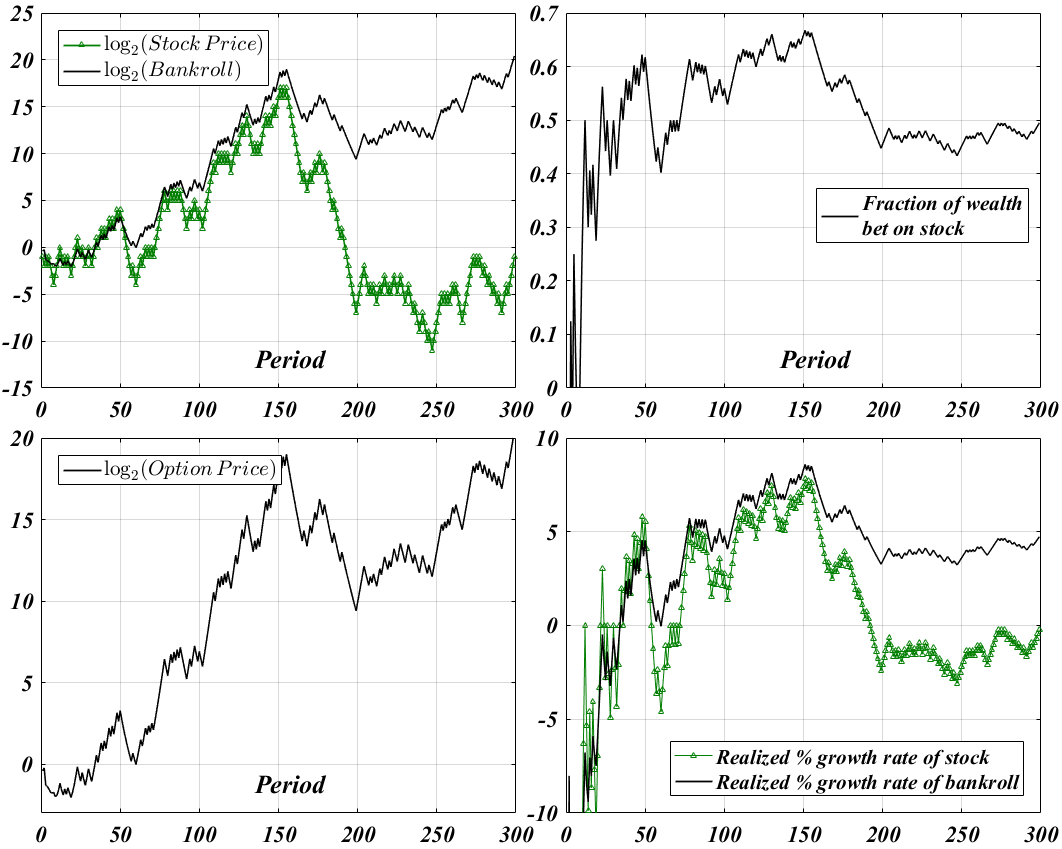}
\caption{\sc Replication of Cover's Derivative on the binomial lattice for Shannon's canonical example (``Shannon's Demon'').}
\label{fig:demon}
\end{figure}
\section{Several Underlyings}
We turn our attention to the general stock market with $n$ correlated stocks $(i=1,...,n)$ that follow the geometric Brownian motions
\begin{equation}
\frac{dS_{it}}{S_{it}}=\mu_i\,dt+\sigma_i\, dW_{it},
\end{equation}where $S_{it}$ is the price of stock $i$ at $t$ and $\mu_i,\sigma_i$ are the drift and volatility of stock $i$, respectively.  The $(W_{it})_{i=1}^n$ are standard Brownian motions, with $\rho_{ij}:=\text{Corr}(dW_{it},dW_{jt})$ being the correlation coefficient of the instantaneous changes in $W_{it}$ and $W_{jt}$.  The correlation matrix, which is assumed to be invertible, is denoted $R:=[\rho_{ij}]_{n\times n}$.
Next, we let
\begin{equation}
\sigma_{ij}:=\rho_{ij}\sigma_i\sigma_j=\text{Cov}\bigg(\frac{dS_{it}}{S_{it}},\frac{dS_{jt}}{S_{jt}}\bigg)\big/dt.  
\end{equation}
We let $\Sigma:=[\sigma_{ij}]_{n\times n}$ denote the covariance matrix of instantaneous returns per unit time, and we write $\Sigma=MRM$, where $M:=\text{diag}(\sigma_1,...,\sigma_n)$ is the diagonal matrix of volatilities.
\par
We take up the general rebalancing rules $b:=(b_1,...,b_n)'\in\mathbb{R}^n,$ where $b_i$ is the fraction of wealth bet on stock $i$ over the interval $[t,t+dt].$  Thus, the trader keeps the fraction $1-\sum\limits_{i=1}^nb_i$ of his wealth in bonds over the interval $[t,t+dt].$  As before, we let $V_t(b)$ denote the gambler's wealth at $t$, where $V_0:=1.$  The trader's wealth evolves according to
\begin{multline}
\frac{dV_t(b)}{V_t(b)}=\sum\limits_{i=1}^nb_i\frac{dS_{it}}{S_{it}}+\bigg(1-\sum\limits_{i=1}^nb_i\bigg)\frac{dB_t}{B_t}\\=\bigg[\sum\limits_{i=1}^nb_i\mu_i+\bigg(1-\sum\limits_{i=1}^nb_i\bigg)r\bigg]dt+\sum\limits_{i=1}^nb_i\sigma_idW_{it}.\\
\end{multline}
For brevity, let $\mu:=(\mu_1,...,\mu_n)'$ denote the vector of drifts.  We then have
\begin{equation}
\frac{dV_t(b)}{V_t(b)}=[r+(\mu-r\textbf{1})'b]dt+\sum_{i=1}^nb_i\sigma_idW_{it},
\end{equation}where $\textbf{1}:=(1,...,1)'$ is an $n\times 1$ vector of ones.  The solution of this stochastic differential equation is given by
\begin{equation}\label{soln2}
V_t(b)=\exp\bigg\{[r+(\mu-r\textbf{1})'b-b'\Sigma b/2]t+\sum\limits_{i=1}^nb_i\sigma_iW_{it}\bigg\}.
\end{equation}
This can be verified directly by applying the multivariate version of It\^{o}'s Lemma (Bj\"{o}rk 1998) to the function
\begin{equation}
F(W_1,...,W_n,t):=\exp\bigg\{[r+(\mu-r\textbf{1})'b-b'\Sigma b/2]t+\sum\limits_{i=1}^nb_i\sigma_iW_{i}\bigg\}.
\end{equation}
Indeed, we get
\begin{equation}
dF_t=\frac{\partial{F}}{\partial{t}}dt+\sum_{i=1}^n\frac{\partial{F}}{\partial{W_i}}dW_{it}+\frac{1}{2}\sum_{i=1}^n\sum_{j=1}^n\frac{\partial^2F}{\partial W_i\partial W_j}\rho_{ij}dt.
\end{equation}
Substituting $\partial F/\partial t=F\cdot [r+(\mu-r\textbf{1})'b-b'\Sigma b/2]$, $\partial F/\partial W_i=F\cdot b_i\sigma_i,$ and $\partial^2 F/\partial W_i\partial W_j=F\cdot b_ib_j\sigma_i\sigma_j$ yields the desired result.  Proceeding as before, we take the expression
\begin{equation}
\sigma_iW_{it}=\log(S_{it}/S_{i0})-(\mu_i-\sigma_i^2/2)t
\end{equation} and substitute it into (\ref{soln2}).  This yields
\begin{equation}
V_t(b)=\exp\bigg\{(r-b'\Sigma b/2)t+\sum_{i=1}^n b_i[\log(S_{it}/S_{i0})-(r-\sigma_i^2/2)t]\bigg\}.
\end{equation}
For brevity, let 
\begin{equation}
\boxed{z_{i}:=\frac{\log(S_{it}/S_{i0})-(r-\sigma_i^2/2)t}{\sigma_i\sqrt{t}}}.
\end{equation}
Under the equivalent martingale measure, the variables $z:=(z_1,...,z_n)'$ are all unit normals, with correlation matrix $R=[\rho_{ij}]$.  Thus, we can write
\begin{equation}
V_t(b)=\exp\{(r-b'\Sigma b/2)t+\sqrt{t}\cdot z'Mb\}.
\end{equation}
Maximizing $V_t(b)$ with respect to $b$, we get the first-order condition
\begin{equation}
t\Sigma b=\sqrt{t}Mz.
\end{equation}  For simplicity, let $S:=(S_1,...,S_n)'$ denote the vector of stock prices, and let $b(S,t)$ denote the best rebalancing rule in hindsight over $[0,t]$.  Solving the first-order condition yields
\begin{equation}
\boxed{b(S,t)=\frac{1}{\sqrt{t}}\cdot M^{-1}R^{-1}z}.
\end{equation}  The final wealth that accrues to a $\$1$ deposit into the best rebalancing rule in hindsight over $[0,t]$ is
\begin{equation}
\boxed{V_t^*=\exp(rt+z'R^{-1}z/2)=\exp(rt+t\cdot b'\Sigma b/2)}.
\end{equation}
Hence, the final payoff of Cover's Derivative is $V_T^*=\exp(rT+z'R^{-1}z/2)$.  Again, we see that the final wealth of the best (levered) rebalancing rule in hindsight is \textit{Markovian}: it depends only on the current state $(S_1,...,S_n,t)$.
\par
We pass to the multivariate version of the Black-Scholes equation (Wilmott 2001), which governs the no-arbitrage price of ``rainbow'' or ``correlation'' options dependent on several underlyings.  As usual $C(S_1,...,S_n,t)=C(S,t)$ will denote the price of Cover's Derivative.  We solve the differential equation
\begin{equation}
\frac{1}{2}\sum_{i=1}^n\sum_{j=1}^n\rho_{ij}\sigma_i\sigma_jS_iS_j\frac{\partial^2C}{\partial S_i\partial S_j}+r\sum\limits_{i=1}^nS_i\frac{\partial C}{\partial S_i}+\frac{\partial C}{\partial t}-rC=0
\end{equation}
with the boundary condition $C(S,T):=V_T^*(S)=\exp(rT+z_T'R^{-1}z_T/2)$.  As usual, we do this by computing the expected discounted payoff with respect to the equivalent martingale measure.
\par
To this end, we again write
\begin{equation}
z_T=\sqrt{t/T}\cdot z_t+\sqrt{1-t/T}\cdot y,
\end{equation}where 
\begin{equation}
y_i:=\frac{\log(S_{iT}/S_{it})-(r-\sigma_i^2/2)(T-t)}{\sigma_i\sqrt{T-t}}.
\end{equation}The $y_i$ are all unit normals with respect to the equivalent martingale measure $\mathbb{Q}$ and the information available at $t$.  $R$ is the correlation matrix of the random vector $y:=(y_1,...,y_n)'$.  The conditional density of $y$ is
$f(y):=(2\pi)^{-n/2}\det(R)^{-1/2}\exp(-y'R^{-1}y/2)$.  Expanding the quadratic form $z_T'R^{-1}z_T$, we get
\begin{equation}
z_T'R^{-1}z_T/2=\frac{t}{2T}z_t'R^{-1}z_t+\frac{\sqrt{t(T-t)}}{T}z_t'R^{-1}y+\frac{T-t}{2T}y'R^{-1}y.
\end{equation}
Thus, we find that $\mathbb{E}_t^\mathbb{Q}[\exp(z_T'R^{-1}z_T/2)]=$
\begin{multline}
(2\pi)^{-n/2}\det(R)^{-1/2}\exp\bigg(\frac{t}{2T}z_t'R^{-1}z_t\bigg)\\\times\int\limits_{-\infty}^\infty\cdot\cdot\cdot\int\limits_{-\infty}^{\infty}\exp\bigg(-\frac{t}{2T}y'R^{-1}y+\frac{\sqrt{t(T-t)}}{T}z_t'R^{-1}y\bigg)dy_1\cdot\cdot\cdot dy_n.\\
\end{multline}  To evaluate the multiple integral, we use the general formula
\begin{equation}
\int\limits_{-\infty}^\infty\cdot\cdot\cdot\int\limits_{-\infty}^\infty \exp(-y'Ay+\beta'y)dy_1\cdot\cdot\cdot dy_n=\pi^{n/2}\det(A)^{-1/2}\exp(\beta'A^{-1}\beta/4),
\end{equation}where $A$ is any symmetric positive definite $n\times n$ matrix and $\beta=(\beta_1,...,\beta_n)'$ is any vector of constants.  Putting $A:=t/(2T)\cdot R^{-1}$, $\beta:=\sqrt{t(T-t)}\big/T\cdot R^{-1}z_t$, and simplifying, we get
\begin{equation}
\mathbb{E}_t^\mathbb{Q}[\exp(z_T'R^{-1}z_T/2)]=(T/t)^{n/2}\exp(z_t'R^{-1}z_t/2).
\end{equation}
\begin{theorem}
For levered hindsight optimization (over all $b\in\mathbb{R}^n$), the price of Cover's Derivative is
\begin{equation}
\boxed{C(S,t)=(T/t)^{n/2}\exp(rt+z'R^{-1}z/2)=(T/t)^{n/2}\exp(rt+t\cdot b'\Sigma b/2)=(T/t)^{n/2}\cdot V_t^*},
\end{equation}
where $z_i:=\{\log(S_{it}/S_{i0})-(r-\sigma_i^2/2)t\}/(\sigma_i\sqrt{t})$, $b(S,t)$ is the best rebalancing rule in hindsight over $[0,t]$, and $V_t^*$ is the intrinsic value at time $t$.
\end{theorem}
\begin{theorem}
For the general market with $n$ correlated stocks in geometric Brownian motion, the American-style version of Cover's Derivative (that expires at $T$, has zero exercise price, and pays $V_t^*$ upon exercise at $t$) will never be excercised early in equilibrium.  The American price $C_a(S,t)$ is equal to the European price $C_e(S,t)=(T/t)^{n/2}\cdot V_t^*$.
\end{theorem}
\begin{proof}
Immediately, we see that the option is ``worth more alive than dead'' on account of the inequalities $C_a(S,t)\geq C_e(S,t)=(T/t)^{n/2}\cdot V_t^*>V_t^*$ for $0<t<T$.
\end{proof}To find the replicating strategy, we again differentiate the price, getting
\begin{equation}
\boxed{\Delta_i:=\frac{\partial C}{\partial S_i}=\frac{C\cdot(R^{-1}z_t)_i}{S_i\sigma_i\sqrt{t}}=\frac{C\cdot b_i(S,t)}{S_i}},
\end{equation}where $(R^{-1}z_t)_i$ is the $i^{th}$ coordinate of the vector $R^{-1}z_t$.  Thus, we have the relation $\Delta_iS_i/C=b_i(S,t)$.
\begin{theorem}
The replicating strategy for Cover's Derivative bets the fraction $b_i(S,t)$ of wealth on stock $i$ in state $(S,t)$.  Thus, to replicate Cover's Derivative, one just uses the best rebalancing rule in hindsight as it is known at time $t$.
\end{theorem}
For the general stock market, we have again concluded that the following three trading strategies are identical:
\begin{enumerate}
\item
The strategy that looks back over the known price history $[0,t]$, finds the best continuously-rebalanced portfolio in hindsight, and uses that portfolio over the interval $[t,t+dt]$
\item
The $\Delta$-hedging strategy induced by Cover's Derivative
\item
The natural estimator $\Sigma^{-1}(\hat{\mu}-r\textbf{1})$ of the continuous-time Kelly rule $\Sigma^{-1}(\mu-r\textbf{1})$ (cf. Luenberger 1998).
\end{enumerate}
\section{Simulations}
We proceed to give three simulations that help visualize the behavior of the replicating strategy over $T:=200$ years under a risk-free rate of $r:=0.02$.  We let $\nu_i:=\mu_i-\sigma^2_i/2$ denote the compound-annual growth rate of stock $i$, and we normalize the initial stock prices to $S_{i0}:=1$.  We also normalize the trader's initial wealth to $\$1$.  Simulations 1 and 2 deal with the univariate case.  For the first 5 years of the experiment, the trader holds a single share of the stock.  Then at $t=5,$ he puts all his money into Cover's Derivative.  The waiting period is necessary because $C\to+\infty$ as $t\to0^+$.  Thus, for $t\leq5$ the trader's wealth is $S_t$, and for $t\geq5$ his wealth is $S_5C(S_t,t)/C(S_5,5)$.  
\subsubsection*{Simulation 1}
We put $\nu:=0.04$ and $\sigma:=0.7$.  The Kelly growth rate (Luenberger 1998) for this market is $9.17\%$ and the Kelly bet is $b^*=0.54$.  The replicating strategy learns to hold significant cash balances and ``live off the fluctuations,'' which are substantial on account of the $70\%$ annual volatility.  Figure \ref{fig:sim1} gives a sample path.
\begin{figure}[t]
\centering
\includegraphics[width=450px]{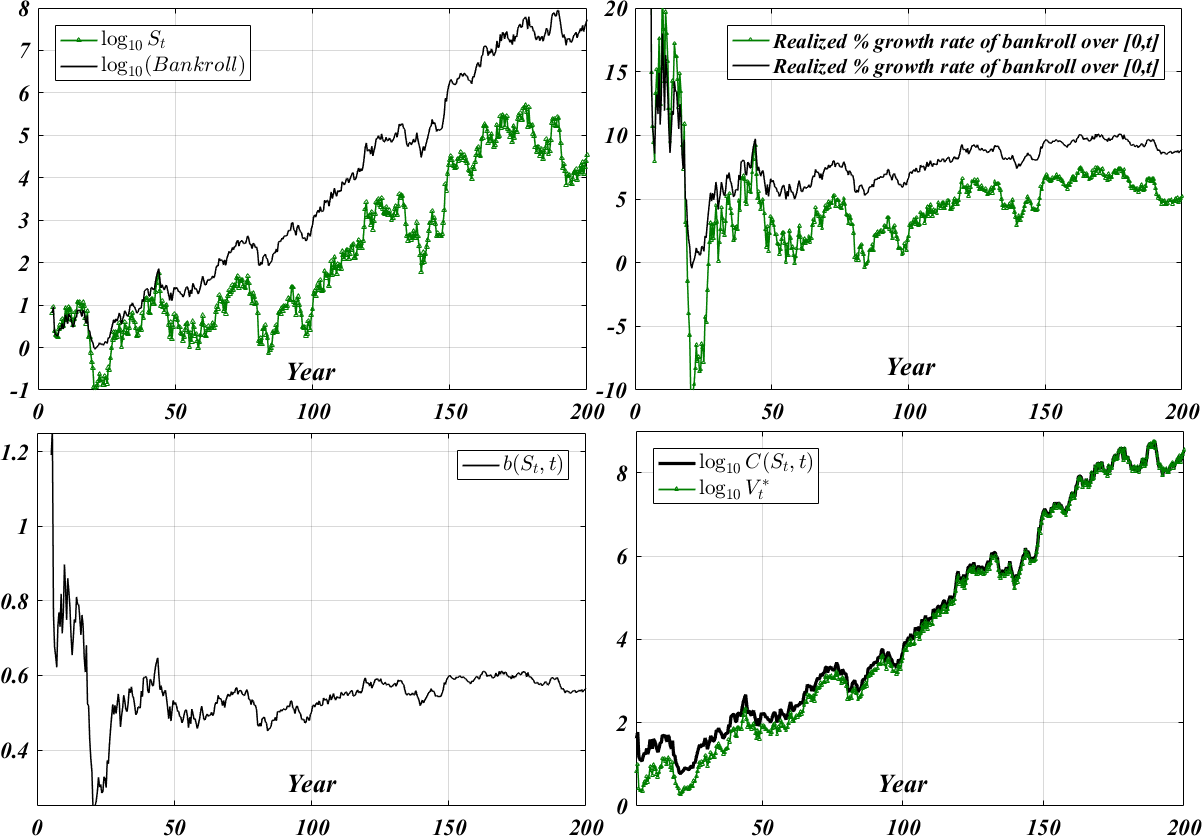}
\caption{\sc{Outcome for $\nu:=0.04,\sigma:=0.7.$}}
\label{fig:sim1}
\end{figure}
\subsubsection*{Simulation 2}
Next, we use $\nu:=0.08$ and $\sigma:=0.17$.  The Kelly growth rate is $11.6\%$ and the Kelly bet is $b^*=2.57$.  The replicating strategy uses enormous leverage in an effort to exploit low interest rates and the favorable risk/return profile.  This is Figure \ref{fig:sim2}.  After 200 years, the stock price has appreciated from $\$1$ a share to $\$100$ million a share, but the replicating strategy has grown the initial dollar into $\$1$ trillion.
\begin{figure}[t]
\centering
\includegraphics[width=450px]{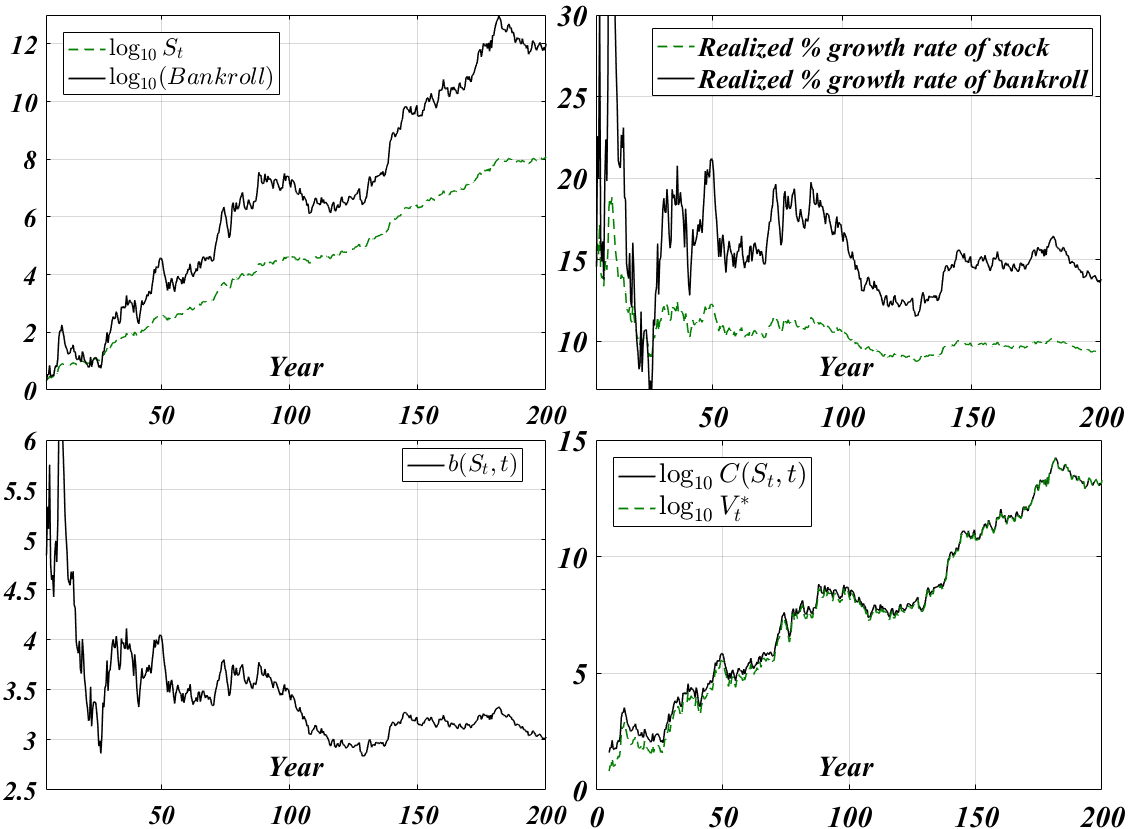}
\caption{\sc Outcome for $\nu:=0.08, \sigma:=0.17.$}
\label{fig:sim2}
\end{figure}
\subsubsection*{Simulation 3}
Finally, we simulate the bivariate case.  At $t=0$, the trader puts $\$0.50$ into each stock.  He holds this portfolio for 5 years, and then he puts all his money into Cover's Derivative.  Thus, for $t\leq5$ his wealth is $0.5[S_{1}(t)+S_{2}(t)]$, and for $t\geq5$ his wealth is $0.5[S_{1}(5)+S_{2}(5)]C(S_t,t)/C(S_5,5)$.
\par
We use $\nu:=(0.03,0.08)'$ and $\sigma:=(0.55,0.7)'$, with $\rho:=0.2$ being the correlation of instantaneous returns.  The Kelly growth rate is $13.7\%$ and the Kelly fractions are $b^*=(0.39,0.56)'$.  Figure \ref{fig:sim3} gives the result.  On this particular sample path, the replicating strategy uses leverage for decades on end, in spite of the fact that a Kelly gambler would continuously hold $5\%$ of wealth in cash.
\begin{figure}[t]
\centering
\includegraphics[width=450px]{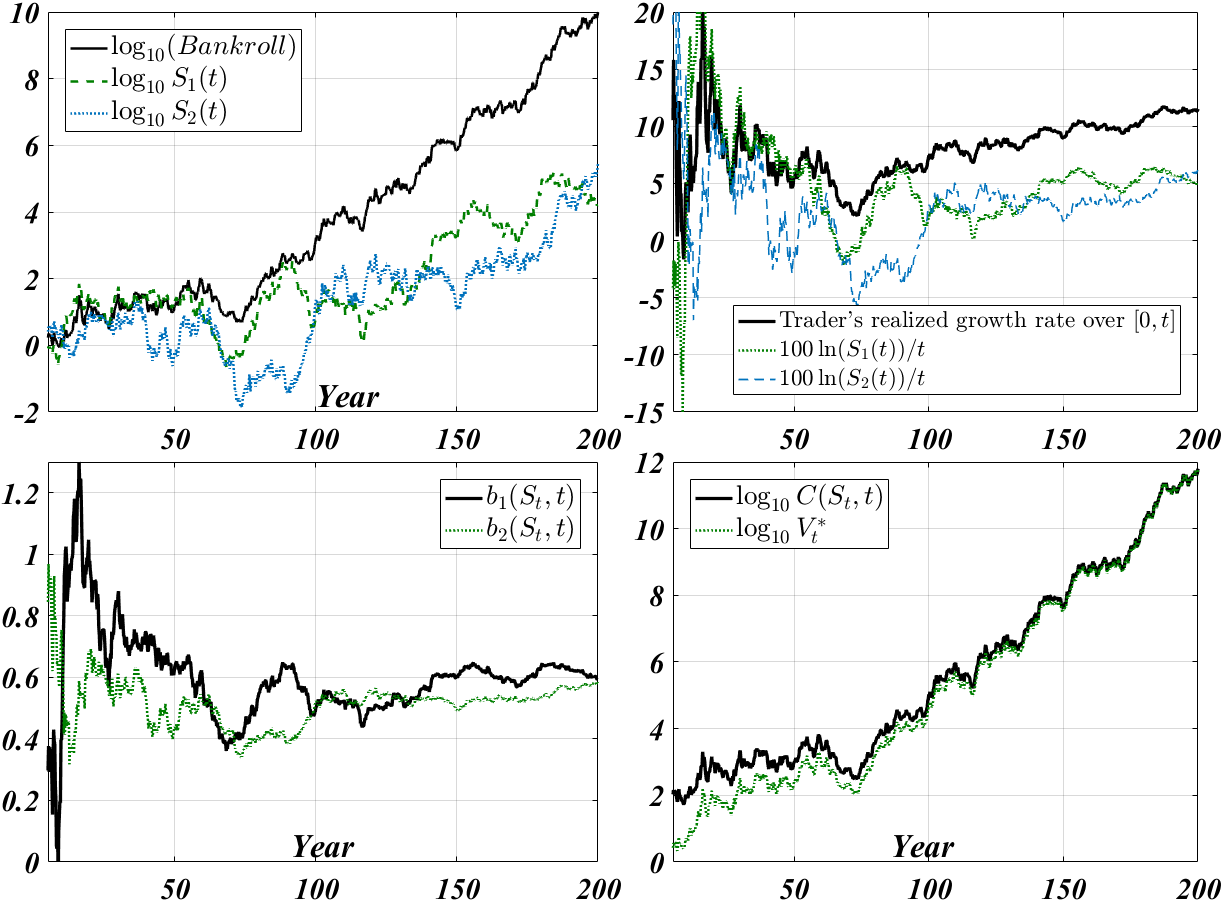}
\caption{\sc Outcome for $\nu:=(0.03,0.08)', \,\sigma:=(0.55,0.7)', \,\rho:=0.2.$}
\label{fig:sim3}
\end{figure}
\section{Limitations and Constraints}
To close the paper, we briefly review the practical limitations of our framework and main results.  First, our entire analysis resides squarely within the Black-Scholes complete market consisting of a risk-free bond and several correlated stocks in geometric Brownian motion.  Accordingly, we have operated under the classical assumption of continuous trading in a frictionless environment that is free of taxes, transaction costs, and bid-ask spreads.  Importantly, we have assumed that one can operate without any institutional constraints on leverage or value at risk.  A potentially unlimited supply of margin loans is presumed to be available at the risk-free rate, and all securities can be sold short with full use of the proceeds.  In accordance with the Kelly theory of asymptotic capital growth, the trader is willing to stomach any level of volatility or short-term drawdown in service of achieving the optimum asymptotic growth rate.  As far as actual praxis on behalf of long lived institutions (say, sovereign wealth funds or university endowments), we have not modelled or simulated the behavior of our investment strategy in the presence of ongoing deposits and withdrawals. 
\par
Finally, we make the technical note that in certain (mathematically degenerate) situations, our ``beat the market asymptotically'' slogan can turn into ``tie the market asymptotically.''  For, if the best rebalancing rule in hindsight over $[0,T]$ amounts to buying and holding one of the stocks (e.g. if $b(S_T,T)$ is a unit basis vector), then the compound growth rate of the practitioner will lag the best performing stock in the market (in the one-underlying case, the market itself) by an amount that becomes vanishingly small as $T\to\infty$.  This finite-sample growth rate lag is precisely the ``cost of universality,'' that is, the cost of having to learn the growth optimal rebalancing rule on-the-fly.
\section{Conclusion}
This paper priced and replicated an exotic option (``Cover's Derivative'') whose payoff equals the final wealth that would have accrued to a $\$1$ deposit into the best leveraged, continuously-rebalanced portfolio in hindsight.  A rebalancing rule is a fixed-fraction betting scheme that trades continuously so as to maintain a target proportion $b_i$ of wealth in each stock $i$.  For the Black-Scholes market with $n$ correlated stocks in geometric Brownian motion, the no-arbitrage price of Cover's rebalancing option is $C(S,t)=(T/t)^{n/2}\exp(rt+t\cdot b'\Sigma b/2)$, where $b=b(S,t)$ is the best rebalancing rule in hindsight over $[0,t]$ and $\Sigma$ is the covariance of instantaneous returns per unit time.  Since $C$ is equal to $(T/t)^{n/2}$ times intrinsic value, the American-style version of Cover's Derivative will never be exercised early in equilibrium because  the option is "worth more alive than dead."
\par
The order of magnitude $C(S,t;T)=\mathcal{O}(T^{n/2})$ agrees with the super-replicating price derived by Cover in his discrete-time universal portfolio theory.  A sophisticated, long-lived institution that puts money into the replicating strategy (a strategy which turns out to be \textit{horizon-free}) will grow its endowment at the same asymptotic rate as the best levered rebalancing rule in hindsight.  In the long-run, with probability approaching 1, it will beat the market averages by an exponential factor.  Of course, this guarantee is subject to the proviso that the best levered rebalancing rule in hindsight must sustain a higher asymptotic growth rate than the market index.
\par
The replicating strategy amounts to betting the fraction $b_i(S,t)$ of wealth on stock $i$ at time $t$, where $b(S,t)$ is the best rebalancing rule in hindsight over the currently known price history.  If someone knows the covariance $\Sigma$ of instantaneous returns (but not necessarily the drifts of the various stocks), he can use the formula $b(S,t)=M^{-1}R^{-1}z/\sqrt{t}$, where $R$ is the correlation matrix of instantaneous returns, $M:=\text{diag}(\sigma_1,...,\sigma_n)$ is the (diagonal) matrix of volatilities, and $z_i:=\{\log(S_{it}/S_{i0})-(r-\sigma_i^2/2)t\}/(\sigma_i\sqrt{t})$.  But even if he is ignorant of $\Sigma$, he can still find $b(S,t)$ at any given time by hindsight-optimizing over the known price history.
\par
Another expression for the replicating strategy is $b(S,t)=\Sigma^{-1}(\hat{\mu}-r\textbf{1})$, where $\hat{\mu}_i:=\log(S_{it}/S_{i0})/t+\sigma_i^2/2$ is the natural estimator of the drift of stock $i$.  The replicating strategy converges in mean square to the continuous-time Kelly rule, $b^*:=\Sigma^{-1}(\mu-r\textbf{1})$, and its realized compound-growth rate converges to the Kelly (1956) optimum asymptotic growth rate, which is $\gamma^*:=r+(1/2)(\mu-r\textbf{1})'\Sigma^{-1}(\mu-r\textbf{1})$.  This happens because the intrinsic value of Cover's Derivative grows at an asymptotic rate of $\gamma^*$ per unit time.  A $\$1$ deposit into the replicating strategy at time $t$ guarantees that the trader will achieve, at $T$, the \textit{deterministic} fraction $V_T^*/C(S_t,t)$ of the final wealth of the best rebalancing rule in hindsight.  The excess continuously-compounded growth rate of $V_T^*$ (over and above that of the replicating strategy) is at most $\{rt+z_t'R^{-1}z_t/2+n\log(T/t)/2\}/(T-t)$, which tends to $0$ as $T\to\infty$.
\section*{Acknowledgment}
I thank Erik Ordentlich and Thomas Cover for their timeless paper, \textit{The Cost of Achieving the Best Portfolio in Hindsight}, which I found uplifting to the spirit.


\begin{thebibliography}{9}
\bibitem{}
\textbf{Barron, A. and Yu, J., 2003}.   Maximal Compounded Wealth for Portfolios of Stocks and Options.  Working Paper.
\bibitem{}
\textbf{Bj\"{o}rk, T., 1998}. \textit{Arbitrage Theory in Continuous Time}. Oxford University Press.
\bibitem{}
\textbf{Black, F. and Scholes, M., 1973}. The Pricing of Options and Corporate Liabilities. \textit{Journal of Political Economy}, \textit{81}(3), pp.637-654.
\bibitem{}
\textbf{Cover, T.M., 1991}. Universal Portfolios. \textit{Mathematical Finance, 1}(1), pp.1-29.
\bibitem{}
\textbf{Cover, T.M. and Gluss, D.H., 1986.} Empirical Bayes Stock Market Portfolios. \textit{Advances in Applied Mathematics, 7}(2), pp.170-181.
\bibitem{}
\textbf{Cover, T.M. and Ordentlich, E., 1996.} Universal Portfolios with Side Information. \textit{IEEE Transactions on Information Theory, 42}(2), pp.348-363.
\bibitem{}
\textbf{Cox, J.C., Ross, S.A. and Rubinstein, M., 1979}. Option Pricing: A Simplified Approach. \textit{Journal of Financial Economics, 7}(3), pp.229-263.
\bibitem{}
\textbf{Cuchiero, C., Schachermayer, W. and Wong, T.K.L., 2016}.  Cover's Universal Portfolio, Stochastic Portfolio Theory and the Num\'eraire Portfolio. \textit{arXiv preprint, arXiv:1611.09631}.
\bibitem{}
\textbf{DeMarzo, P., Kremer, I. and Mansour, Y., 2006},  Online Trading Algorithms and Robust Option Pricing.  In \textit{Proceedings of the Thirty-Eighth Annual ACM Symposium on Theory of Computing}, pp.477-486.
\bibitem{}
\textbf{Fernholz, E.R., 2002}.  \textit{Stochastic Portfolio Theory}. Springer.
\bibitem{}
\textbf{Gort, C., and Burgener, E., 2014}.  Rebalancing Using Options.  Working Paper.
\bibitem{}
\textbf{Gy\"{o}rfi, L., Lugosi, G. and Udina, F., 2006}.  Nonparametric Kernel-Based Sequential Investment Strategies.   \textit{Mathematical Finance 16}(2), pp.337-357.
\bibitem{}
\textbf{Ilmanen, A. and Maloney, T., 2015.}  Portfolio Rebalancing Part 1 of 2:  Strategic Asset Allocation.  AQR White Paper, 2015.
\bibitem{}
\textbf{Israelov, R. and Tummala, H., 2018.}  An Alternative Option to Portfolio Rebalancing.  \textit{The Journal of Derivatives, (25)}3, pp.7-32.
\bibitem{}
\textbf{Iyengar, G., 2005}.  Universal Investment in Markets with Transaction Costs.  \textit{Mathematical Finance, 15}(2), pp.359-371.
\bibitem{}
\textbf{Jamshidian, F., 1992}. Asymptotically Optimal Portfolios. \textit{Mathematical Finance, 2}(2), pp.131-150.
\bibitem{}
\textbf{Kelly, J.L., 1956}. A New Interpretation of Information Rate. \textit{Bell System Technical Journal, 35}(4), pp.917-926.
\bibitem{}
\textbf{Kozat, S.S. and Singer, A.C., 2011}.  Universal Semiconstant Rebalanced Portfolios.  \textit{Mathematical Finance,  21}(2), pp.293-311.
\bibitem{}
\textbf{Luenberger, D.G., 1998}. \textit{Investment Science}. Oxford University Press.
\bibitem{}
\textbf{Merton, R.C., 1973}. Theory of Rational Option Pricing. \textit{The Bell Journal of Economics and Management Science}, pp.141-183.
\bibitem{}
\textbf{Ordentlich, E. and Cover, T.M., 1998.} The Cost of Achieving the Best Portfolio in Hindsight. \textit{Mathematics of Operations Research}, \textit{23}(4), pp.960-982.
\bibitem{}
\textbf{Parkes, D.C. and Huberman, B.A., 2001}.  Multiagent Cooperative Search for Portfolio Selection.  \textit{Games and Economic Behavior, 35}, pp.124-165.
\bibitem{}
\textbf{Poundstone, W., 2010}.  \textit{Fortune's Formula: The Untold Story of the Scientific Betting System that Beat the Casinos and Wall Street}.   Hill and Wang.
\bibitem{}
\textbf{Reiner, E. and Rubinstein, M., 1992}. Exotic Options. Working paper.
\bibitem{}
\textbf{Rujeerapaiboon, N., Kuhn, D. and Wiesemann, W., 2015}.  Robust Growth-Optimal Portfolios.  \textit{Management Science, 62}(7), pp.2090-2109.
\bibitem{}
\textbf{Stoltz, G. and Lugosi, G., 2005}. Internal Regret in On-Line Portfolio Selection.  \textit{Machine Learning, 59}(1-2), pp.125-159.
\bibitem{}
\textbf{Wilmott, P., 1998.} \textit{Derivatives:  the Theory and Practice of Financial Engineering}.  John Wiley \& Sons.
\bibitem{}
\textbf{Wilmott, P., 2001}. \textit{Paul Wilmott Introduces Quantitative Finance}. John Wiley \& Sons.
\bibitem{}
\textbf{Wong, T.K.L., 2015}.  Universal Portfolios in Stochastic Portfolio Theory. \textit{arXiv preprint, arXiv:1510.02808}.
\end{thebibliography}
\end{document}